\documentclass[11pt,twocolumn]{llncs}
\usepackage{fullpage}
\usepackage{amsmath}
\usepackage{amssymb}
\usepackage{graphicx}
\usepackage{subfigure}
\usepackage{algorithm}
\usepackage{algpseudocode}
\usepackage{caption}
\usepackage{lipsum}
\usepackage{color}
\usepackage[margin=1in, paperwidth=8.5in, paperheight=11in]{geometry}

\usepackage{fancyhdr}
\newcommand{\removed}[1]{}

\newcommand{\cp}{\mathcal{P}}

\newcommand{\ra}{\rightarrow}
\newcommand{\rsa}{\rightsquigarrow}
\newcommand{\coleq}{\mathrel{\mathop:}=}
\newcommand{\bs}{\backslash}

\newcommand{\future}{\mathrm{future}}
\newcommand{\arrival}{\mathrm{arrival}}

\title{Naming and Counting in Anonymous Unknown Dynamic Networks\thanks{This work has in part been supported by the EU (European Social Fund - ESF) and Greek national funds through the Operational Programme ``Education and Lifelong Learning'' (EdLL), under the title ``Foundations of Dynamic Distributed Computing Systems'' (\textsf{FOCUS}).}}

\author{Othon Michail \and Ioannis Chatzigiannakis \and Paul G. Spirakis}
\institute{Computer Technology Institute \& Press ``Diophantus'' (CTI), Patras, Greece
\\Email:\email{ \{michailo, ichatz, spirakis\}@cti.gr}}

\begin{document}

\maketitle

\noindent\textbf{Abstract.} 
In this work, we study the fundamental naming and counting problems (and some variations) in networks that are anonymous, unknown, and possibly dynamic. In \emph{counting}, nodes must determine the size of the network $n$ and in \emph{naming} they must end up with unique identities. By \emph{anonymous} we mean that all nodes begin from identical states apart possibly from a unique leader node and by \emph{unknown} that nodes have no a priori knowledge of the network (apart from some minimal knowledge when necessary) including ignorance of $n$. Network dynamicity is modeled by the 1-interval connectivity model \cite{KLO10}, in which communication is synchronous and a worst-case adversary chooses the edges of every round subject to the condition that each instance is connected. We first focus on static networks with broadcast where we prove that, without a leader, counting is impossible to solve and that naming is impossible to solve even with a leader and even if nodes know $n$. These impossibilities carry over to dynamic networks as well. We also show that a unique leader suffices in order to solve counting in linear time. Then we focus on dynamic networks with broadcast. We conjecture that dynamicity renders nontrivial computation impossible. In view of this, we let the nodes know an upper bound on the maximum degree that will ever appear and show that in this case the nodes can obtain an upper bound on $n$. Finally, we replace broadcast with \emph{one-to-each}, in which a node may send a different message to each of its neighbors. Interestingly, this natural variation is proved to be computationally equivalent to a full-knowledge model, in which unique names exist and the size of the network is known.

\section{Introduction}
\label{sec:intro}

Distributed computing systems are more and more becoming dynamic. The static and relatively stable models of computation can no longer represent the plethora of recently established and rapidly emerging information and communication technologies. In recent years, we have seen a tremendous increase in the number of new mobile computing devices. Most of these devices are equipped with some sort of communication, sensing, and mobility capabilities. Even the Internet has become mobile. The design is now focused on complex collections of heterogeneous devices that should be robust, adaptive, and self-organizing, possibly moving around and serving requests that vary with time. Delay-tolerant networks are highly-dynamic, infrastructure-less networks whose essential characteristic is a possible absence of end-to-end communication routes at any instant. Mobility can vary from being completely predictable to being completely unpredictable. Gossip-based communication mechanisms, e-mail exchanges, peer-to-peer networks, and many other contemporary communication networks all assume or induce some sort of highly-dynamic communication network. 

The formal study of dynamic communication networks is hardly a new area of research. There is a huge amount of work in distributed computing that deals with causes of dynamicity such as failures and changes in the topology that are rather slow and usually eventually stabilize (like, for example, in self-stabilizing systems \cite{Do00}). However the low rate of topological changes that is usually assumed there is unsuitable for reasoning about truly dynamic networks. Even graph-theoretic techniques need to be revisited: the suitable graph model is now that of a \emph{dynamic graph} (a.k.a. \emph{temporal graph} or \emph{time-varying graph}) (see e.g. \cite{KKK00,Ko09,CFQS11}), in which each edge has an associated set of time-labels indicating availability times. Even fundamental properties of classical graphs do not carry over to their temporal counterparts. See, for example, \cite{KKK00} for a violation of Menger's theorem and \cite{AKL08} for the unsuitability of the standard network diameter metric. 

In this work, we adopt as our dynamic network model the \emph{$1$-interval connectivity} model that was proposed in the seminal STOC paper of Kuhn \emph{et al.} \cite{KLO10} building upon previous work of O'Dell and Wattenhofer \cite{OW05}. In this model, nodes proceed in \emph{synchronous rounds} and communicate by \emph{interchanging messages}. Message transmission is \emph{broadcast} in which, in every round, each node issues a single message to be delivered to all its neighbors. In this model, the network may change arbitrarily from round to round subject to the condition that in each round the network is connected. We only consider deterministic algorithms.

We focus on networks in which nodes are initially identical and, unless necessary, do not have any information about the network. In any case, nodes do not know the size $n$ of the network. By \emph{identical} we mean that they do not have unique identities (ids) and execute identical programs. So, this is some sort of minimal reliable distributed system, like, for example, a collection of particularly cheap and bulk-produced wireless sensor nodes. Nodes may execute the same program, because it is too costly to program them individually and their lack of ids may be due to the fact that ids require customization beyond the capabilities of mass production \cite{AFFR06}. Our only assumption is the existence of a unique leader that introduces some symmetry breaking. To further break the symmetry introduced by broadcast message transmission and in order to solve naming in dynamic networks, we allow to the nodes to send a different message to each one of their neighbors. 

\section{Related Work}
\label{sec:rw}

Distributed systems with worst-case dynamicity were first studied in \cite{OW05}. Their outstanding novelty was to assume a communication network that may change arbitrarily from time to time subject to the condition that each instance of the network is connected. They studied asynchronous communication and allowed nodes detect local neighborhood changes. They studied the \emph{flooding} and \emph{routing} problems in this setting and among others provided a uniform protocol for flooding that terminates in $O(Tn^2)$ rounds using $O(\log n)$ bit storage and message overhead, where $T$ is the maximum time it takes to transmit a message.

Computation under worst-case dynamicity was further and extensively studied in a series of works by Kuhn \emph{et al.} in the synchronous case. In \cite{KLO10}, among others, \emph{counting} (in which nodes must determine the size of the network) and \emph{all-to-all token dissemination} (in which $n$ different pieces of information, called tokens, are handed out to the $n$ nodes of the network, each node being assigned one token, and all nodes must collect all $n$ tokens) were solved in $O(n^2)$ rounds using $O(\log n)$ bits per message. Several variants of \emph{coordinated consensus} in 1-interval connected networks were studied in \cite{KOM11}. Requiring continuous connectivity has been supported by the findings of \cite{CKLL09}, where a \emph{connectivity service} for mobile robot swarms that encapsulates an arbitrary motion planner and can refine any plan to preserve connectivity while ensuring progress was proposed. 

Some recent works \cite{Ha11,HK11} present information spreading algorithms in worst-case dynamic networks based on \emph{network coding}. An \emph{open} setting in which nodes constantly join and leave has very recently been considered in \cite{APRU12}. For an excellent introduction to distributed computation under worst-case dynamicity see \cite{KO11}. Two very thorough surveys on dynamic networks are \cite{Sc02,CFQS11}.

The question concerning which problems can be solved by a distributed system when all processors use the same algorithm and start from the same state has a long story with its roots dating back to the seminal work of Angluin \cite{An80}, who investigated the problem of establishing a ``center''. She was the first to realize the connection with the theory of graph coverings, which was going to provide, in particular with the work of Yamashita and Kameda \cite{YK96}, several characterizations for problems that are solvable under certain topological constraints. Further investigation led to the classification of computable functions \cite{YK96,ASW88}. \cite{BV99} removed the, until then, standard assumption of knowing the network size $n$ and provided characterizations of the relations that can be
computed with arbitrary knowledge. Other well-known studies on unknown networks have dealt with the problems of robot-exploration and map-drawing of an unknown graph \cite{AH00,DP90,PP98} and on information dissemination \cite{AGVP90}. Sakamoto \cite{Sa99} studied the ``usefulness'' of initial conditions for distributed algorithms (e.g. leader or knowing $n$) on anonymous networks by presenting a transformation algorithm from one initial condition to another. Fraigniaud \emph{et al.} \cite{FPP00} assumed a unique leader in order to break symmetry and assign short labels as fast as possible. To circumvent the further symmetry introduced by broadcast message transmission they also studied other natural message transmission models as sending only one message to a single neighbor. Recently, and independently of our work, Chalopin \emph{et al.} \cite{CMM12} have studied the problem of naming anonymous networks in the context of snapshot computation. Finally, Aspnes \emph{et al.} \cite{AFFR06} studied the relative powers of reliable anonymous distributed systems with different communication mechanisms: anonymous broadcast, read-write registers, or read-write registers plus additional shared-memory objects.

\section{Contribution}
\label{sec:con}

We begin, in Section \ref{sec:prel}, by formally describing our distributed models. In Section \ref{sec:prob}, we formally define the problems under consideration, that is, naming, counting and some variations of these. Our study begins, in Section \ref{sec:static}, from static networks with broadcast. The reason for considering static networks is to arrive at some impossibility results that also carry over to dynamic networks, as a static network is a special case of a dynamic network. In particular, we prove that naming is impossible to solve under these assumptions even if a unique leader exists and even if all nodes know $n$. Then we prove that without a leader also counting is impossible to solve and naturally, in the sequel, we assume the existence of a unique leader. We provide an algorithm based on the \emph{eccentricity} of the leader (greatest distance of a node from the leader) that solves counting in linear time (inspired by the findings in \cite{FPP00}). Then, in Section \ref{sec:dynbr}, we move on to dynamic networks with broadcast. We begin with a conjecture (and give evidence for it) essentially stating that dynamicity renders nontrivial computations impossible even in the presence of a unique leader. \footnote{By \emph{nontrivial computation} we mean the ability to decide any language $L$ on input assignments s.t. $L\neq \Sigma^*$ and $L\neq\emptyset$, where input symbols are chosen from some alphabet $\Sigma$. For example, deciding the existence of any symbol in the input is considered nontrivial.} In view of this, we allow the nodes some minimal initial knowledge, which is an upper bound on the maximum degree that any instance will ever have. This could for example be some natural constraint on the capacity of the network. We provide a protocol that exploits this information to compute an upper bound on the size of the network. However, w.r.t. naming, the strong impossibility from Section \ref{sec:static} still persists (after all, knowledge of $n$ does not help in labeling the nodes). To circumvent this, in Section \ref{sec:seltr}, we relax our message transmission model to \emph{one-to-each} that allows each node to send a different message to each one of its neighbors. This is an alternative communication model that has been considered in several important works, like \cite{Ha11}, however in different contexts than ours. This further symmetry breaking, though minimal, allows us, by exploiting a leader, to uniquely label the nodes. By this, we establish that this model is equivalent to a full-knowledge model in which unique names exist and the size of the network is known. To arrive at this result, we provide four distinct naming protocols each with its own incremental value. The first presents how to assign ids in a fair context in which the leader will eventually meet every other node. The second improves on the first by allowing all nodes to assign ids in a context where no one is guaranteed to meet everybody else, but where connectivity guarantees progress. Both these are correct stabilizing solutions that do not guarantee termination. Then we provide a third protocol that builds upon the first two and manages to assign unique ids in 1-interval connected graphs while terminating in linear time. As its drawback is that messages may be $\Omega(n^2)$ bit long, we refine it to a more involved fourth protocol that reduces the bits per message to $\Theta(\log n)$ by only paying a small increase in termination time. 

\section{Preliminaries}
\label{sec:prel}

\subsection{The models}
\label{subsec:mod}

A \emph{dynamic network} is modeled by a \emph{dynamic graph} $G=(V,E)$, where $V$ is a set of $n$ nodes (or processors) and $E:\bbbn\ra \cp(E^\prime)$, where $E^\prime=\{\{u,v\}:u,v\in V\}$, (wherever we use $\bbbn$ we mean $\bbbn_{\geq 1}$) is a function mapping a round number $r\in\bbbn$ to a set $E(r)$ of bidirectional links drawn from $E^\prime$. Intuitively, a dynamic graph $G$ is an infinite sequence $G(1),G(2),\ldots$ of \emph{instantaneous graphs}, whose edge sets are subsets of $E^\prime$ chosen by a \emph{worst-case adversary}. A \emph{static network} is just a special case of a dynamic network in which $E(i+1)=E(i)$ for all $i\in\bbbn$. The set $V$ is assumed throughout this work to be \emph{static}, that is it remains the same throughout the execution.

A dynamic graph/network $G=(V,E)$ is said to be \emph{$1$-interval connected}, if, for all $r\in\bbbn$, the static graph $G(r)$ is connected \cite{KLO10}. Note that this allows the network to change arbitrarily from round to round always subject to the condition that it remains connected. In this work, we focus on $1$-interval connected dynamic networks which also implies that we deal with connected networks in the static-network case.

Nodes in $V$ are \emph{anonymous} that is they do not initially have any ids and they do not know the topology or the size of the network, apart from some minimal knowledge when necessary (i.e. we say that \emph{the network is unknown}). However, nodes have unlimited local storage. In several cases, and in order to break symmetry, we may assume a unique \emph{leader node} (or \emph{source}) $l$. If this is the case, then we assume that $l$ starts from a unique initial state $l_0$ (e.g. 0) while all other nodes start from the same initial state $q_0$ (e.g. $\perp$). All nodes but the leader execute identical programs. 

Communication is \emph{synchronous message passing} \cite{Ly96,AW04}, meaning that it is executed in discrete rounds controlled by a global clock that is available to the nodes and that nodes communicate by sending and receiving messages. Thus all nodes have access to the current round number via a local variable that we usually denote by $r$. We consider two different models of message transmission. One is \emph{anonymous broadcast}, in which, in every round $r$, each node $u$ generates a single message $m_u(r)$ to be delivered to all its current neighbors in $N_u(r)=\{v:\{u,v\}\in E(r)\}$. The other is \emph{one-to-each} in which a different message $m_{(u,i)}(r)$, $1\leq i\leq d_u(r)$, where $d_u(r)\coleq|N_u(r)|$ is the degree of $u$ in round $r$, may be generated for each neighbor $v_i$. In every round, the adversary first chooses the edges for the round; for this choice it can see the internal states of the nodes at the beginning of the round. In the one-to-each message transmission model we additionally assume that the adversary also reveals to each node $u$ a set of locally unique edge-labels $1,2,\ldots,d_u(r)$, one for each of the edges currently incident to it. Note that these labels can be reselected arbitrarily in each round so that a node cannot infer what the internal state of a neighbor is based solely on the corresponding local edge-name. Then each node transitions to a new state based on its internal state (containing the messages received in the previous round) and generates its messages for the current round: in anonymous broadcast a single message is generated and in one-to-each a different message is generated for each neighbor of a node. Note that, in both models, a node does not have any information about the internal state of its neighbors when generating its messages. Deterministic algorithms are only based on the current internal state to generate messages. This implies that the adversary can infer the messages that will be generated in the current round before choosing the edges. Messages are then delivered to the corresponding neighbors. In one-to-each, we assume that each message $m_i$ received by some node $u$ is accompanied with $u$'s local label $i$ of the corresponding edge, so that a node can associate a message sent through edge $i$ with a message received from edge $i$. These messages will be processed by the nodes in the subsequent round so we typically begin rounds with a ``receive'' command referring to the messages received in the previous round. Then the next round begins.   

\subsection{Causal Influence}

Probably the most important notion associated with a dynamic graph is the \emph{causal influence}, which formalizes the notion of one node ``influencing'' another through a chain of messages originating at the former node and ending at the latter (possibly going through other nodes in between). We use $(u,r)\rsa (v, r^\prime)$ to denote the fact that node $u$'s state in round $r$ ($r$-state of $u$) influences node $v$'s state in round $r^\prime$. Formally:

\begin{definition} [\cite{La78}]
Given a dynamic graph $G=(V,E)$ we define an order $\ra\subseteq (V\times\bbbn_{\geq 0})^2$, where $(u,r)\ra (v,r+1)$ iff $u=v$ or $\{u,v\}\in E(r+1)$. The \emph{causal order} $\rightsquigarrow\subseteq (V\times\bbbn_{\geq 0})^2$ is defined to be the reflexive and transitive closure of $\ra$.
\end{definition}

A very important aspect of 1-interval connectivity, that will be invoked in all our proof arguments in dynamic networks, is that it guarantees that the state of a node causally influences the state of another uninfluenced node in every round (if one exists). To get an intuitive feeling of this fact, consider a partitioning of the set of nodes $V$ to a subset $V_1$ of nodes that know the $r$-state of some node $u$ and to a subset $V_2=V\backslash V_1$ of nodes that do not know it. Connectivity asserts that there is always an edge in the cut between $V_1$ and $V_2$, consequently, if nodes that know the $r$-state of $u$ broadcast it in every round, then in every round at least one node moves from $V_2$ to $V_1$.

This is formally captured by the following lemma from \cite{KLO10}.

\begin{lemma} [\cite{KLO10}] \label{lem:inf}
For any node $u\in V$ and $r\geq 0$ we have
\begin{enumerate}
\item $|\{v\in V : (u,0)\rsa (v,r)\}|\geq\min\{r + 1, n\}$,
\item $|\{v\in V : (v,0)\rsa (u,r)\}|\geq\min\{r + 1, n\}$.
\end{enumerate}
\end{lemma}

\section{Problem Definitions}
\label{sec:prob}

\noindent\textbf{$k$-labeling}. An algorithm is said to solve the $k$-labeling problem if whenever it is executed on a network comprising $n$ nodes each node $u$ eventually terminates and outputs a \emph{label} (or \emph{name} or \emph{id}) $id_u$ so that $|\{id_u: u\in V\}|\geq k$.

\vspace{0.5cm}

\noindent\textbf{Naming}. The naming problem is a special case of the $k$-labeling problem in which it must additionally hold that $k=n$. This, in turn, implies that $id_u\neq id_v$ for all distinct $u,v\in V$ (so, unique labels are required for the nodes).

\vspace{0.5cm}

\noindent\textbf{Minimal (Consecutive) Naming}. It is a special case of naming in which it must additionally hold that the $n$ nodes output the labels $\{0,1,\ldots,n-1\}$.

\vspace{0.5cm}

\noindent\textbf{Counting Upper Bound}. Nodes must determine an upper bound $k$ on the network size $n$.

\vspace{0.5cm}

\noindent\textbf{Counting}. A special case of counting upper bound in which it must hold that $k=n$.  

\section{Static networks with broadcast}
\label{sec:static}

We here assume that the network is described by a static graph $G=(V,E)$, where $E\subseteq\{\{u,v\}:u,v\in V\}$. Moreover, the message transmission model is broadcast, that is, in every round, each node $u$ generates a single message to be delivered to all its neighbors. Note that any impossibility result established for static networks is also valid for dynamic networks as a static network is a special case of a dynamic network.

First of all, note that if all nodes start from the same initial state then, if we restrict ourselves to deterministic algorithms, naming is impossible to solve in general static networks, even if nodes know $n$. The reason is that in the worst-case they may be arranged in a ring (in which each node has precisely 2 neighbors) and it is a well-known fact \cite{An80,Ly96,AW04} that, in this case, in every round $r$, all nodes are in identical states.

We show now that impossibility persists even if we allow a unique leader and even if nodes have complete knowledge of the network. 

\begin{theorem} \label{the:namimp}
Naming is impossible to solve by deterministic algorithms in general anonymous (static) networks with broadcast even in the presence of a leader and even if nodes have complete knowledge of the network.
\end{theorem}
\begin{proof}
Consider a star graph with the leader in the center (see Appendix \ref{app:star}).
\qed
\end{proof}

An obvious generalization is that, under the same assumptions as in the statement of the theorem, it is impossible to solve $k$-labeling for any $k\geq 3$. In Appendix \ref{app:deg}, we also provide some thoughts on a degree-based labeling.

We now turn our attention to the simpler counting problem. First we establish the necessity of assuming a unique leader.

\begin{theorem} \label{the:impcoun}
Without a leader, counting is impossible to solve by deterministic algorithms in general anonymous networks with broadcast.
\end{theorem}
\begin{proof}
If some algorithm counts in $k$ rounds the $n$ nodes of a static ring, then it fails on a ring of $k+1$ nodes (see Appendix \ref{app:ring}).
\qed
\end{proof}

In view of Theorem \ref{the:impcoun}, we assume again a unique leader in order to solve counting. Recall that the \emph{eccentricity} of a node $u$ is defined as the greatest geodesic distance between $u$ and $v$, over all $v\in V\backslash\{u\}$, where ``distance'' is equivalent to ``shortest path''. We first describe a protocol \textit{Leader\_Eccentricity} (inspired by the $Wake\& Label$ set of algorithms of \cite{FPP00}) that assigns to every node a label equal to its distance from the leader and then we exploit this to solve counting. We assume that all nodes have access to the current round number via a variable $r$.

\noindent\textbf{Protocol \textit{Leader\_Eccentricity}.} The leader begins with $label\leftarrow 0$ and $max\_asgned\leftarrow 0$ and all other nodes with $label\leftarrow\perp$. In the first round, the leader broadcasts an $assign$ $(1)$ message. Upon reception of an $assign$ $(i)$ message, a node that has $label=\perp$ sets $label\leftarrow i$ and broadcasts to its neighbors an $assign$ $(i+1)$ message and an $ack$ $(i)$ message. Upon reception of an $ack$ $(i)$ message, a node with $label\neq\perp$ and $label<i$ broadcasts it. Upon reception of an $ack$ $(i)$ message, the leader sets $max\_asgned\leftarrow i$ and if $r > 2\cdot (max\_asgned+1)$ then it broadcasts a $halt$ message, outputs its label, and halts. Upon reception of a $halt$ message, a node broadcasts $halt$, outputs its label, and halts.

\begin{theorem}
In $Leader\_Eccentricity$ nodes output $\epsilon+1$ distinct labels where $\epsilon$ is the eccentricity of the leader. In particular, every node outputs its distance from the leader.
\end{theorem}
\begin{proof}
At time $2$, nodes at distance $1$ from the leader receive $assign$ $(1)$ and set their label to $1$. By induction on distance, nodes at distance $i$ get label $i$ at round $i+1$. In the same round, they send an ack that must arrive at the leader at round $2i+1$. If not then there is no node at distance $i$.
\qed
\end{proof}

We now use $Leader\_Eccentricity$ to solve counting in anonymous unknown static networks with a leader. We additionally assume that at the end of the $Leader\_Eccentricity$ process each node $u$ knows the number of neighbors $up(u)=|\{\{v,u\}\in E: label(v)=label(u)-1\}|$ it has to its upper level (it can store this during the $Leader\_Eccentricity$ process by counting the number of $assign$ messages that arrived at it from its upper level neighbors). Moreover, we assume that all nodes know the leader's eccentricity $\epsilon$ (just have the leader include $max\_asgned$ in its $halt$ message). Finally, let, for simplicity, the first round just after the completion of the above process be round $r=1$. For this, we just need all nodes to end concurrently the $Leader\_Eccentricity$ process. This is done by having node with label $i$ that receives or creates (this is true for the leader) a $halt$ message in round $r$ halt in round $(r+max\_asgned-i)$. Then the nodes just reset their round counters. 

\noindent\textbf{Protocol \textit{Anonymous\_Counting}.} Nodes first execute the modified $Leader\_Eccentricity$. 
When $\epsilon-r+1=label(u)$, a non-leader node $u$ receives a possibly empty (in case of no lower-level neighbors) set of $partial\_count_i$ $(rval_i)$ messages and broadcasts a $partial\_count$ $((1+\sum_i rval_i)/up(u))$ message. When $r=\epsilon+1$, the leader receives a set of $partial\_count_i$ $(rval_i)$ messages, sets $count\leftarrow 1+\sum_i rval_i$, broadcasts a $halt$ $(count)$ message, outputs $count$, and halts. When a non-leader $u$ receives a $halt$ $(count)$ message, it outputs $count$ and halts.

For a given round $r$ we denote by $rval_i(u)$ the $i$th message received by node $u$. 

\begin{theorem}
$Anonymous\_Counting$ solves the counting problem in anonymous static networks with broadcast under the assumption of a unique leader. All nodes terminate in $O(n)$ rounds and use messages of size $O(\log n)$.
\end{theorem}
\begin{proof}
By induction on the round number $r$, in the beginning of round $r\geq 2$, it holds that $\sum_{u:label(u)=\epsilon-r+1} \left ( 1+\sum_{i} rval_i(u)\right )=|\{u:label(u)\geq \epsilon -r+1\}|$. Clearly, in round $\epsilon+1$ it holds that $count=1+\sum_i rval_i(leader)=|\{u:label(u)\geq 0\}|=n$.
\qed
\end{proof}

\section{Dynamic Networks with Broadcast}
\label{sec:dynbr}

We now turn our attention to the more general case of 1-interval connected dynamic networks with broadcast. We begin with a conjecture stating that dynamicity renders nontrivial computation impossible (evidence for this conjecture can be found in Appendix \ref{app:conj}; see also \cite{OW05} for a similar conjecture in a quite different setting). Then we naturally strengthen the model to allow some computation.

\begin{conjecture} \label{conj:pred}
It is impossible to compute (even with a leader) the predicate $N_a\geq 1$, that is ``exists an $a$ in the input'', in general anonymous unknown dynamic networks with broadcast.
\end{conjecture}

In view of Theorem \ref{the:namimp}, which establishes that we cannot name the nodes of a static, and thus also of a dynamic, network if broadcast communication is assumed, and of the above conjecture, implying that in dynamic networks we cannot count even with a leader \footnote{This is implied because if we could count we could have a node wait at most $n-1$ rounds until it hears of an $a$ (provided that all nodes that have heard of an $a$ forward it) and if no reject.}, we start strengthening our initial model.

Let us assume that there is a unique leader $l$ that knows an upper bound $d$ on maximum degree ever to appear in the dynamic network, that is $d\geq\max_{u\in V,r\in\bbbn}\{d_u(r)\}$. We keep the broadcast message transmission.

Note first that impossibility of naming persists. However, we show that obtaining an upper bound on the size of the network now becomes possible, though exponential in the worst case.

\noindent\textbf{Protocol \textit{Degree\_Counting}.} The leader stores in $d$ the maximum degree that will ever appear and begins with $label\leftarrow 0$, $count\leftarrow 1$, $latest\_event\leftarrow 0$, $max\_label\leftarrow 0$, and $r\leftarrow 0$ while all other nodes begin with $label\leftarrow\perp$, $count\leftarrow 0$, and $r\leftarrow 0$. In the beginning of each round each node increments by one its round counter $r$. The leader in each round $r$ broadcasts $assign$ $(r)$. Upon reception of an $assign$ $(r\_label)$ message, a node with $label=\perp$ sets $label\leftarrow r\_label$ and from now in each round $r$ broadcasts $assign$ $(r)$ and $my\_label$ $(label)$. A node with $label=\perp$ that did not receive an $assign$ message sends an $unassigned$ $(r)$ message. All nodes continuously broadcast the maximum $my\_label$ and $unassigned$ messages that they have received so far. Upon reception of an $unassigned$ $(i)$ message, the leader, if $i>latest\_event$, it sets $count\leftarrow 1$ and, for $k=1,\ldots,i$, $count\leftarrow count+d\cdot count$, $max\_label\leftarrow i$, and $latest\_event\leftarrow r$ and upon reception of a $my\_label$ $(j)$ message, if $j>max\_label$, it sets $count\leftarrow 1$ and, for $k=1,\ldots,j$, $count\leftarrow count+d\cdot count$, $latest\_event\leftarrow r$, and $max\_label\leftarrow j$ (if receives both $i,j$ it does it for $\max\{i,j\}$). When it holds that $r> count+latest\_event-1$ (which must eventually occur) then the leader broadcasts a $halt$ $(count)$ message for $count$ rounds and then outputs $count$ and halts. Each node that receives a $halt$ $(r\_count)$ message, sets $count\leftarrow r\_count$, broadcasts a $halt$ $(count)$ message for $count$ rounds and then outputs $count$ and halts. 

\begin{theorem}
$Degree\_Counting$ solves the counting upper bound problem in anonymous dynamic networks with broadcast under the assumption of a unique leader. The obtained upper bound is $O(d^n)$ (in the worst case). 
\end{theorem}
\begin{proof}
In the first round, the leader assigns the label $1$ to its neighbors and obtains an $unassigned$ $(1)$ message from each one of them. So, it sets $count\leftarrow (d+1)$ (in fact, note that in the first step it can simply set $count\leftarrow d_u(1)+1$, but this is minor), $latest\_event\leftarrow 1$, and $max\_label\leftarrow 1$. Now, if there are further nodes, at most by round $count+latest\_event-1$ it must have received an $unassigned$ $(i)$ message with $i>latest\_event$ or a $my\_label$ $(j)$ with $j>max\_label$. Note that the reception of an $unassigned$ $(i)$ message implies that at least $i+1$ distinct labels have been assigned because as long as there are unlabeled nodes one new label is assigned in each round to at least one node (this is implied by Lemma \ref{lem:inf} and the fact that all nodes with labels constantly assign new labels). Initially, one node (the leader) assigned to at most $d$ nodes label $1$. Then the $d+1$ labeled nodes assigned to at most $(d+1)d$ unlabeled nodes the label $2$, totalling $(d+1)+(d+1)d$, and so on.

In the worst-case, each label in $\{0,1,\ldots,n-1\}$ is assigned to precisely one node (e.g., consider a static line with the leader in the one endpoint). In this case the nodes count $O(d^n)$.
\qed
\end{proof}

We point out that if nodes have access to more drastic initial knowledge such as an upper bound $e$ on the \emph{maximum expansion}, defined as $\max_{u,r,r^\prime}\{|\future_{u,r}(r^\prime +1)|-|\future_{u,r}(r^\prime)|\}$ (maximum number of concurrent new influences ever occuring), where $\future_{(u,r)}(r^\prime)\coleq \{v\in V : (u,r)\rsa (v,r^{\prime})\}$, for $r\leq r^\prime$, then essentially the same protocol as above provides an $O(n\cdot e)$ upper bound. 

\section{Dynamic Networks with One-to-Each}
\label{sec:seltr}

The result of Theorem 1, in the light of (a) the conjecture of Section 7, and (b) the assumption of a broadcast message transmission model, clearly indicates that nontrivial computations in anonymous unknown dynamic networks are impossible even under the assumption of a unique leader. We now relax these assumptions so that we can state a correct naming protocol. We start by relaxing the assumption of a broadcast message transmission medium by offering to nodes access to a \emph{one-to-each message transmission mechanism}. We also assume a unique leader; without a leader, even under a one-to-each model, impossibility still persists as any pair of nodes that form a static ring will not be able to break symmetry.

%======================================
\subsection*{1st Version - Protocol $Fair$}
%======================================

We now present protocol $Fair$ in which the unique leader assigns distinct labels to each node of the network. The labels assigned are tuples $(r,h,i)$, where $r$ is the round during which the label was assigned, $h$ is the label of the leader node and $i$ is a unique number assigned by the leader. The labels can be uniquely ordered first by $r$, then by $h$ and finally by $i$ (in ascending order). 

Each node maintains the following local variables: $clock$, for counting the rounds of execution of the protocol (implemented due to synchronous communication, see Sec. \ref{subsec:mod}), $label$, for storing the label assigned by the leader, $state$, for storing the local state that can be set to $\{anonymous, named, leader\}$, and $counter$, for storing the number of labels generated. All nodes are initialized to $clock\leftarrow 0$, $id\leftarrow (0,\perp,\perp)$, $state\leftarrow anonymous$, and $counter\leftarrow 0$ except from the leader that is initialized to $clock\leftarrow 0$, $id\leftarrow (0,1,1)$, $state\leftarrow leader$, and $counter\leftarrow 1$. 

Each turn the leader $u$ consults the one-to-each transmission mechanism and identifies a set of locally unique edge-labels $1,2,\ldots,d(u)$, one for each of the edges incident to it. \footnote{Recall from Section 4.1 that these edge-labels can be reselected arbitrarily in each round (even if the neighbors remain the same) by the adversary so that a node cannot infer what the internal state of a neighbor is, based solely on the corresponding local edge-name.} The leader iterates the edge-label set and transmits to each neighboring node a different message $m_i$, $1 \le i \le d(u)$ that contains the unique label $(clock, label, counter + i)$. When the transmission is complete, it increases the variable $counter$ by $d(u)$. All the other nodes of the network do not transmit any messages (or transmit a null message if message transmission is compulsory).

All nodes under $state=anonymous$, upon receiving a (non-null) message set the local $label$ to the contents of the message and change $state$ to $named$. All the other nodes of the network simply ignore all the messages received.

At the end of the turn all nodes do $clock++$ (where `++' is interpreted as ``increment by one'').

Recall that a naming assignment is correct if \emph{all nodes} are assigned \emph{unique} labels. It is clear that $Fair$ is a non-terminating correct protocol, given the following \emph{fairness assumption}: the leader node at some point has become directly connected with each other node of the network (i.e., eventually meets all nodes).

\begin{lemma}
With one-to-each transmission, under the fairness assumption, and in the presence of a unique leader, protocol $Fair$ eventually computes a unique assignment for all the nodes in any anonymous unknown dynamic network.
\end{lemma}

%======================================
\subsection*{2nd Version - Protocol $Delegate$}
%======================================

We now proceed by presenting a stronger protocol $Delegate$ (based on $Fair$) that is correct even without the fairness assumption. To achieve correctness the leader node delegates the role of assignment of labels to all the nodes that it encounters. Thus, without loss of generality, even if the leader does not encounter all other nodes of the network, due to the \emph{connectivity property}, all nodes will eventually hear from the leader. Therefore, all nodes will either receive a unique label from the leader or from another labeled node. The uniqueness among the labels generated is guaranteed since each label can be traced back to the node that issued it using the $h$ parameter.

In $Delegate$ the nodes maintain the same variables as in $Fair$. Each turn the leader performs the same actions as in $Fair$. Also similarly to $Fair$, each node that is in $state=anonymous$ does not transmit any message (or transmits a null message if message transmission is compulsory). Each node $u$ that is in $state=named$ performs similar actions as the leader node and transmits to each edge-label $i$ a message containing the unique label $(clock_u, label_u, counter_u + i)$ and then increases the variable $counter_u$ by $d(u)$. All nodes under $state=anonymous$, upon receiving one or more (non-null) messages that contain a label, select the message that contains the lowest label (i.e., the one with the lowest $h$ parameter) and set the local $label$ to the contents of the message and change $state$ to $named$. At the end of the turn all nodes do $clock++$.

\begin{lemma}
With one-to-each transmission, and in the presence of a unique leader, protocol $Delegate$ correctly computes a unique assignment for all the nodes in any anonymous unknown dynamic network.
\end{lemma}

%======================================
\subsection*{3rd Version - Protocol $Dynamic\_Naming$ (terminating)}
%======================================

The protocols $Fair$ and $Delegate$ compute a correct naming assignment (based on different assumptions) but do not terminate. Essentially the nodes continue to transmit labels for ever. We now present protocol $Dynamic\_Naming$ (based on $Delegate, Fair$) that manages to terminate.

$Dynamic\_Naming$ is an $O(n)$-time protocol that assigns unique ids to the nodes and informs them of $n$. As usual, there is a unique leader $l$ with id $0$ while all other nodes have id $\perp$.

The idea here is as follows. Similarly to $Delegate$, all nodes that have obtained an id assign ids and these ids are guaranteed to be unique. Additionally to $Delegate$, we have nodes that have obtained an id to acknowledge their id to the leader. Thus, all nodes send their ids and all nodes continuously forward the received ids so that they eventually arrive at the leader (simple flooding mechanism). So, at some round $r$, the leader knows a set of assigned ids $K(r)$. We describe now the termination criterion. If $|K(r)|\neq |V|$ then in at most $|K(r)|$ additional rounds the leader must hear (be causally influenced) from a node outside $K(r)$ (to see why, see Lemma \ref{lem:inf}). Such a node, either has an id that the leader first hears of, or has no id yet. In the first case, the leader updates $K(r)$ and in the second waits until it hears of a new id (which is guaranteed to appear in the future). On the other hand, if $|K(r)|=|V|$ no new info will ever arrive at the leader in the future and the leader may terminate after the $|K(r)|$-round waiting period ellapses.

\noindent\textbf{Protocol \textit{Dynamic\_Naming}.} Initially, every node has three variables $count\leftarrow 0$, $acks\leftarrow\emptyset$, and $latest\_unassigned\leftarrow 0$ and the leader additionally has $latest\_new\leftarrow 0$, $time\_bound\leftarrow 1$, and $known\_ids\leftarrow\{0\}$. A node with $id\neq\perp$ for $1\leq i\leq k$ sends $assign$ $(id,count+i)$ message to its $i$th neighbor and sets $count\leftarrow count+k$. In the first round, the leader additionally sets $known\_ids\leftarrow\{0,(0,1),(0,2),\ldots,(0,k)\}$, $latest\_new\leftarrow 1$, and $time\_bound\leftarrow 1+|known\_ids|$. Upon receipt of $l$ $assign$ messages $(rid_j)$, a node with $id=\perp$ sets $id\leftarrow \min_j\{rid_j\}$ (in number of bits), $acks\leftarrow acks\cup id$, sends an $ack$ $(acks)$ message to all its $k$ current neighbors, for $1\leq i\leq k$ sends $assign$ $(id,count+i)$ message to its $i$th neighbor, and sets $count\leftarrow count+k$. Upon receipt of $l$ $ack$ messages $(acks_j)$, a nonleader sets $acks\leftarrow acks\cup (\bigcup_j acks_j)$ and sends $ack$ $(acks)$. A node with $id=\perp$ sends $unassigned$ $(current\_round)$. Upon receipt of $l\geq 0$ $unassigned$ messages $(val_j)$, a node with $id\notin\{0,\perp\}$ sets $latest\_unassigned\leftarrow\max\{latest\_unassigned,\max_j\{val_j\}\}$ and sends $unassigned$ $(latest\_unassigned)$. Upon receipt of $l$ $ack$ messages $(acks_j)$, the leader if $(\bigcup_j acks_j)\backslash known\_ids\neq\emptyset$ sets $known\_ids\leftarrow known\_ids\cup (\bigcup_j acks_j)$, $latest\_new\leftarrow current\_round$ and $time\_bound\leftarrow current\_round+|known\_ids|$ and upon receipt of $l$ $unassigned$ messages $(val_j)$, it sets $latest\_unassigned\leftarrow\max\{latest\_unassigned,\max_j\{val_j\}\}$. If, at some round $r$, it holds at the leader that $r>time\_bound$ and $latest\_unassigned<latest\_new$, the leader sends a $halt$ $(|known\_ids|)$ message for $|known\_ids|-1$ rounds and then outputs $id$ and halts. Any node that receives a $halt$ $(n)$ message, sends $halt$ $(n)$ for $n-2$ rounds and then outputs $id$ and halts.

Denote by $S(r)=\{v\in V:(l,0)\rsa (v,r)\}$ the set of nodes that have obtained an id at round $r$ and by $K(r)$ those nodes in $S(r)$ whose id is known by the leader at round $r$, that is $K(r)=\{u\in V: \exists r^\prime \text{ s.t. } u\in S(r^\prime) \text{ and } (u,r^\prime)\rsa (l,r)\}$. 

\begin{theorem}
$Dynamic\_Naming$ solves the naming problem in anonymous unknown dynamic networks under the assumptions of one-to-each message transmission and of a unique leader. All nodes terminate in $O(n)$ rounds and use messages of size $\Theta(n^2)$.
\end{theorem}
\begin{proof}
Unique names are guaranteed as in $Delegate$. Termination is as follows. Clearly, if $V\backslash K(r)\neq\emptyset$, either $|K(r+|K(r)|)|\geq |K(r)|+1$ or $(u,r)\rsa (l,r+|K(r)|)$ for some $u\in V\backslash S(r)$. The former is recognized by the leader by the arrival of a new id and the latter by the arrival of an $unassigned$ $(timestamp)$ message, where $timestamp\geq r$. On the other hand, if $K(r)=V$ then $|K(r+|K(r)|)|=|K(r)|$ and $\nexists u\in V\backslash S(r)$ s.t. $(u,r)\rsa (l,r+|K(r)|)$ as $V\backslash S(r)=\emptyset$. Finally, note that connectivity implies that $|S(r+1)|\geq \min\{|S(r)|+1,n\}$ which in turn implies $O(n)$ rounds until unique ids are assigned. Then another $O(n)$ rounds are required until nodes terminate.
\qed
\end{proof}

Clearly, by executing a simple $O(n)$-time process after $Dynamic\_Naming$ we can easily reassign minimal (consecutive) names to the nodes. The leader just floods a list of $(old\_id,new\_id)$ pairs, one for each node in the network.  

%======================================
%\section*{4th Version - Protocol $Individual\_Conversations$ (logarithmic messages)}
%======================================

Though $Dynamic\_Naming$ is a correct and time-efficient terminating protocol for the naming problem it still has an important drawback. The messages sent may be of size $\Omega(n^2)$. We now refine $Dynamic\_Naming$ to arrive at a more involved construction that reduces the message size to $\Theta(\log n)$ by paying a small increase in termination time. We call this 4th version of our naming protocols $Individual\_Conversations$. Due to space restrictions, we only give that main idea here. The full presentation can be found in the Appendix \ref{app:indcon}.

\noindent\textbf{Protocol \textit{Individual\_Conversations} [Main Idea].} To reduce the size of the messages (i) the assigned names are now of the form $k\cdot d+id$, where $id$ is the id of the node, $d$ is the number of \emph{unique consecutive} ids that the leader knows so far, and $k\geq 1$ is a name counter (ii) Any time that the leader wants to communicate to a remote node that has a unique id it sends a message with the id of that node and a timestamp equal to the current round. The timestamp allows all nodes to prefer this message from previous ones so that the gain is twofold: the message is delivered and no node ever issues a message containing more than one id. The remote node then can reply in the same way. For the assignment formula to work, nodes that obtain ids are not allowed to further assign ids until the leader freezes all named nodes and reassigns to them unique consecutive ids. During freezing, the leader is informed of any new assignments by the named nodes and terminates if all report that no further assignments were performed.  

\begin{theorem}
$Individual\_Conversations$ solves the (minimal) naming problem in $O(n^3)$ rounds using messages of size $\Theta(\log n)$.
\end{theorem}

Finally, in the Appendix \ref{app:hdy}, we discuss how a high-dynamicity assumption can help in breaking the impossibility of Conjecture \ref{conj:pred} and exploit the above algorithmic ideas to solve naming under broadcast communication.

\newpage

\appendix

\section*{Appendix}

\section{Proof of Theorem \ref{the:namimp}}
\label{app:star}

\textbf{Theorem \ref{the:namimp}.} \emph{Naming is impossible to solve by deterministic algorithms in general anonymous (static) networks with broadcast even in the presence of a leader and even if nodes have complete knowledge of the network.}
\begin{proof}
Imagine a star graph in which the leader has $n-1$ neighbors (it is the \emph{center}) and every other node has only the leader as its unique neighbor (they are the \emph{leaves}). All leaf-nodes are in the same initial state and receive the same first message $m_1$ from the center. So they all transition to the same new state and generate the same outgoing message. It is straightforward to verify, by induction on the number of rounds, that in every round $r$ all leaf-nodes are in identical states. In fact, in any network in which some node is connected to at least two terminal nodes, that is nodes with no further neighbors, those terminal nodes will forever be in identical states.
\qed
\end{proof}

\section{Degree-Labeling}
\label{app:deg}

Consider now again the case of identical nodes (i.e. no leader). Note that, in the beginning of the 2nd round, each node $u$ can determine its degree $d(u)$ by just counting the number of received messages. If we restrict ourselves on networks with at least $k$ different degrees, that is $|\bigcup_{u\in V}d(u)|\geq k$ then we can solve $k$-labeling by simply having each node $u$ output $d(u)$. However, as there is no graph with $n$ different degrees, there is no single network class in which this solution would solve naming. To see that there is no simple graph (not even disconnected) with $n$ different degrees notice that the maximum degree is $n-1$ so the $n$ degrees must be $0,1,\ldots,n-1$. For a node to have degree $n-1$, all other nodes must have at least $1$, thus no node can have 0.

An interesting open question is ``what is the maximum $k$ for which the above algorithm solves $k$-labeling?''. Equivalently, ``what is the maximum number of different degrees that a connected simple graph can have?''.

\section{Proof of Theorem \ref{the:impcoun}}
\label{app:ring}

\textbf{Theorem \ref{the:impcoun}.} \emph{Without a leader, counting is impossible to solve by deterministic algorithms in general anonymous networks with broadcast.}
\begin{proof}
For the sake of contradiction, assume that an algorithm $A$ solves it. Then it solves it on a static ring $R_1$ of size $n$ with the first node terminating in $k\geq n$ rounds. Now consider a ring $R_2$ of size $k+1$. All nodes in both rings are initially in the same identical initial state $\perp$. Thus, any node in $R_2$ has the same $k$-neighborhood (states of nodes in distance at most $k$) as any node in $R_1$ which implies that after $k$ rounds these two nodes will be in the same state (see e.g. Lemma 3.1 in \cite{ASW88}). Thus a node in $R_2$ terminates after $k$ rounds and outputs $n$ which is a contradiction.
\qed
\end{proof}

\section{Evidence for Conjecture \ref{conj:pred}}
\label{app:conj}

The conjecture essentially based on the following fact. Even in a dynamic network, it can be the case that two nodes that are initially in the same state $a$ can for any number of rounds $T$ have the same $T$-neighborhood, which means that the whole history of received messages is the same in both nodes and thus they always transition to identical states. This is, for example, true in a symmetric tree rooted at the leader (e.g. a tree with $k$ identical lines leaving the root) in which the two nodes are in each round in equal distance from the root (even if this distance changes from round to round by moving the 2 nodes back and forth). In dynamic networks, it is also the case that for a node $u$ to causally influence the leader with its $t$-state, all nodes that receive the $t$-state of $u$ should continuously broadcast it at least until the leader receives it (then they could probably stop by receiving an ack or by using some known upper bound on the delivery time). Potentially, $O(n)$ nodes can receive the $t$-state of $u$ before it is delivered to the leader. It seems that if the leader could at some point decide that the received messages originate from two distinct nodes that are forever in identical states then it would also decide the same on a dynamic network containing only one of these nodes, as in both cases the whole network could be full of messages of the same kind. So, it seems impossible for the leader to determine whether the network contains at least two $a$s and such a process is necessary for the leader to count the size of the network. To determine whether there are no $a$s at all, in the absence of $a$s, the leader should somehow determine that it has been causally influenced by the whole network, which in turn requires counting.

\section{4th Version - Protocol $Individual\_Conversations$ (logarithmic messages)}
\label{app:indcon}

Though $Dynamic\_Naming$ is a correct and time-efficient terminating protocol for the naming problem it still has an important drawback. The messages sent may be of size $\Omega(n^2)$. There are two reasons for this increased message size. One is the method of assigning ids, in which the id of a node is essentially set to a pair containing the id of its fisrt parent and a counter. By induction on assignments, in which the leader assigns to a single node, that node assigns to another node, the third node to a fourth one, and so on, it is easy to see that ids may become $n$-tuples and thus have size $O(n)$. The other reason is that, for a node to acknowledge to the leader its assigned id, that node and all nodes that receive it must continuously broadcast it until the leader receives it (otherwise, delivery is not guaranteed by our dynamic network model). As $O(n)$ nodes may want to acknowledge at the same time, it follows that some node may need to continuously broadcast $O(n)$ ids each of size $O(n)$, thus $O(n^2)$. We now refine $Dynamic\_Naming$ to arrive at a more involved construction that reduces the message size to $\Theta(\log n)$ by paying a small increase in termination time. We call this protocol $Individual\_Conversations$. Due to the many low-level details of the protocol we adopt a high-level but at the same time precise and clear verbal presentation.

One refinement concerns the method of assigning ids. We notice that if some $d$ nodes have the unique consecutive ids $D=\{0,1,2,\ldots,k-1\}$, then we can have node with id $j\in D$ assign ids $k\cdot d+j$, for all $k\geq 1$. For example, if we have nodes $\{0,1,2,3\}$, then node 0 assigns ids $\{4,8,12,\ldots\}$, node 1 assigns $\{5,9,13,\ldots\}$, node 2 assigns $\{6,10,14,\ldots\}$, and node 3 assigns $\{7,11,15,\ldots\}$. Clearly, the assignments are unique and in the worst case $k,d,j=O(n)$, which implies that the maximum assigned id is $O(n^2)$ thus its binary representation is $\Theta(\log n)$. So, if we could keep the assigning nodes to have unique consecutive ids while knowing the maximum existing id (so as to evaluate the id-generation formula), we could get logarithmic ids.

Even if we could implement the above assignment method, if nodes continued to constantly forward all ids that they ever hear of then we would not do better than message sizes $O(n\log n)$ (a node forwards $O(n)$ ids each of size $O(\log n)$). Clearly, another necessary improvement is to guarantee communication between the leader and some node with unique id $j$ that the leader knows of, i.e. a pairwise conversation. It is important that a conversation is initiated by the leader so that we do not get multiple nodes trying to initiate a conversation with the leader, as this would increase the communication complexity. The leader sends a $request (rem\_id,current\_round)$ message, where $rem\_id$ is the id of the remote node and $current\_round$ is a timestamp indicating the time in which the request for conversation was initiated. Upon receipt of a $request (r\_id,timestamp)$ message all nodes such that $id\neq r\_id$ forward the message if it is the one with the largest timestamp that they have ever heard of. All nodes keep forwarding the message with the largest timestamp. When the remote node receives the message it replies with $report (id, current\_round)$, where $id$ is its own id. Now all nodes will forward the report as it is the one with the largest timestamp and the report will eventually reach the leader who can reply with another request, and so on. Note that a node that participates in a conversation need not know how much time it will take for the other node to reply. It only needs to have a guarantee that the reply will eventually arrive. Then it can recognize that this is the correct reply by the type, the id-component, and the timestamp of the received message. A nice property of 1-interval connected graphs is that it guarantees any such reply to arrive in $O(n)$ rounds if all nodes that receive it keep broadcasting it (which is the case here, due to the timestamps). So, in order to keep the message sizes low, we must implement the above communication method in such a way that the leader always participates in a single conversation, so that a single message ever floods the whole network (in particular, the most recently created one).

Now, let us further develop our id-assignment method. Clearly, in the 1st round the leader can keep id $0$ for itself and assign the unique consecutive ids $\{1,2,\ldots,d_l(1)\}$ to its $|d_l(1)|$ neighbors in round $1$. Clearly, each node with id $j$ in $K(1)=\{0,1,\ldots,|d_l(1)|\}$ can further assign the unique ids $k\cdot |K(1)|+j$, for $k\geq 1$. As before we can have a node stick to the smallest id that it hears from its neighbors but we additionally need that node to remember those ids that it rejected in a $rejected$ list. However, note that, if nodes outside $K(1)$ that obtain a unique id are not allowed to further assign ids, then we do not guarantee that all nodes will eventually obtain an id. The reason is that the adversary can forever hide the set $K(1)$ from the rest of the graph via nodes that have obtained an id and do not further assign ids (that is, all nodes in $K(1)$ may communicate only to nodes in $K(1)$ and to nodes that have obtained an id but do not assign and all nodes that do not have an id may communicate only to nodes that do not have an id and to nodes that have obtained an id but do not assign, which is some sort of deadlock). So we must somehow also allow to nodes that obtain ids to further assign ids. The only way to do this while keeping our assignment formula is to restructure the new assignments so that they are still unique and additionally consecutive. So, for example, if nodes in $K(1)$ have at some point assigned a set of ids $T$, then the leader should somehow reassign to nodes in $T$ the ids $\{|K(1)|,|K(1)|+1,\ldots,|K(1)|+|T|-1\}$.

So, at this point, it must be clear that the leader must first allow to the nodes that have unique consecutive ids (including itself) to perform some assignments. Then at some point it should freeze the assigning nodes and ask them one after the other to report the assignments that they have performed so far. Then, assuming that it has learned all the newly assigned unique ids, it should communicate with that nodes to reassign to them the next available unique consecutive ids and also it should inform all nodes with id of the maximum consecutive id that has been assigned so far. Now that all nodes with id have unique consecutive ids and know the maximum assigned, they can all safely use the id-assignment formula. In this manner, we have managed to also allow to the new nodes to safely assign unique ids. Finally, the leader unfreezes the nodes with ids one after the other, alows them to assign some new ids and at some point freezes them again to repeat the above process which we may call a cycle. 

A very important point that we should make clear at this point is that, in 1-interval connected graphs, a new assignment is only guaranteed if at least for one round all nodes that have ids send assignment messages to all their incident edges. As if some node with id selected to issue no assignment message to some of its edges then the adversary could make that edge be the only edge that connects nodes with ids to nodes without ids and it could do the same any time an edge is not used. Forunately, this is trivially guaranteed in the solution we have develped so far. When the leader unfreezes the last node with id, even if it chooses to start freezing the nodes in the subsequent round, provided that at least for that round it does not freezes itself, then in that round all nodes including itself are not frozen, thus all take an assignment step in that round (sending assignment messages to all their incident edges). This guarantees that for at least one round all assign at the same time which in turn guarantees at least one new delivery, provided that there are still nodes without ids.  

Another point that is still blur is the following. When the leader gets all reports from all nodes that were assigning ids during this cycle it cannot know which ids have been assigned but only which ids have been possibly assigned. The reason is that when a node $u$ assigns some ids then it is not guaranteed that in the next round it will have the same neighbors. So it can be the case that some of its neighbors chooses to stick to a smaller id sent by some other node and $u$ never notices it. So we have each node that assigns ids to remember the ids that have possibly been assigned and each node that is assigned an id to remember those ids that it rejected. Note that when a node $u$ tries to assigns an id by sending it via a local edge, then, in the next round when it receives from that local edge, it can tell whether that id was possibly assigned by simply having all nodes send their id in every round. If the received id from that edge was $\perp$ then the corresponding neighbor did not have an id thus it must have been assigned some id even if that was not the one sent by $u$. In any case, the id sent by $u$ will either be assigned or stored in the $rejected$ list of that node. On the other hand, if the received id was not equal to $\perp$ then the neighbor already had an id, $u$ knows that its assignment was for sure unsuccessful and may reuse this id in future assignments. The problem now is that, if the leader tries to initiate a conversation with an arbitrary id from those that have been possibly assigned, it can very well be the case that this id was not assigned and the leader may have to wait for a reply forever. Fortunately, this can be solved easily by having the unique node that has stored this id in its $rejected$ list to reply not only when it gets a $request$ message containing its own id but also when it gets a message containing an id that is also in its $rejected$ list. Another way is the following. As the leader has first collected all possibly delivered ids, it can order them increasingly and start seeking that smallest id. As nodes stick to the smallest they hear, the smallest of all possibly assigned was for sure selected by some node. Then that node may inform the leader of some rejected ids which the leader will remove from its ordering and then the leader may proceed to seek for the next id that has remained in its ordered list. It is not hard to see that this method guarantees that the leader always seeks for existing ids. 

Finally, the termination criterion is more or less the same as in $Dynamic\_Naming$. The leader knows that, if it allows all nodes with ids a common assignment step, then, provided that there are nodes without ids, at least one new assignment must take place. Clearly, if all nodes report that they performed no assignments, then the leader can terminate (and tell others to terminate) knowing that all nodes must have obtained an id. In the termination phase, it can reassign for a last time unique consecutive ids and inform all nodes of $n$. 

\begin{theorem}
$Individual\_Conversations$ solves the (minimal) naming problem in anonymous unknown dynamic networks under the assumptions of one-to-each message transmission and of a unique leader. All nodes terminate in $O(n^3)$ rounds and use messages of size $\Theta(\log n)$.
\end{theorem}

\section{Higher Dynamicity}
\label{app:hdy}

Given some high-dynamicity assumption (some sort of fairness), naming can be solved under broadcast communication. Intuitively, to break the symmetry that is responsible for the impossibility of Conjecture \ref{conj:pred}, we require that, given sufficient time, a node has influenced every other node in different rounds. Formally, there must exist $k$ (not necessarily known to the nodes) s.t $(\arrival_{(u,r)}(v),\arrival_{(u,r+1)}(v),\ldots,$ $\arrival_{(u,r+k-1)}(v))\neq$ $(\arrival_{(u,r)}(w),$ $\arrival_{(u,r+1)}(w),\ldots,\arrival_{(u,r+k-1)}(w))$, $\forall u\in$ $V,r\geq 0,v,w\in V\bs\{u\}$, where $\arrival_{(u,r)}(v)$ $\coleq \min\{r^\prime>r :(u,r)\rsa (v,r^\prime)\}$ (first time that $v$ is causally influenced by the $r$-state of $u$). We also allow nodes to have time to acknowledge to their neighbors (formally, we may duplicate each instance of the dynamic graph, i.e. make it persist for two rounds). 

The idea is to have the leader name its first $d_l(1)$ neighbors say with id $1$. What the leader can exploit is that it knows the number of $1$s in the network as it knows its degree in round 1. Now every node $v$ named 1 counts $\arrival_{(l,i)}(v)$ for all $i\geq 2$. This is achieved by having the leader continuously send an $(l,current\_round)$ pair, unnamed nodes constantly forward it, and having every node named $1$ set $\arrival_{(l,i)}(v)$ to the round in which an $(l,i)$ pair was first delivered. It is clear that, due to the above high-dynamicity assumption, the vector $s(v)=(1,\arrival_{(l,2)}(v),\arrival_{(u,3)}(v),\ldots,$ $\arrival_{(u,k+2)}(v))$ (in $k$ rounds) will be a unique $id$. As the named nodes do not know $k$, we have them continuously send $(s,current\_round)$ pairs, where $s$ is the above vector, and all other nodes continuously forward these pairs. At some point, the leader must hear from $d_l(1)$ different $s$ vectors with equal timestamps and then it knows that the $1$s have obtained unique ids. Now the leader can stop them from further changing their ids. Then it allows them (including itself) to concurrently assign id 2 for at least one step. Assigning nodes count the number of assignments that they perform (in a variable $count$ initially $0$). This is done by having a node $u$ that was assigned id 2 in round $r$ to respond to its neighbors the number $l$ of nodes that tried to assigned $2$ to it. Then each of the assigning $1$s sets $count\leftarrow count + 1/l$. When the leader freezes the $1$s, they report their $count$ variable and by summing them the leader learns the number, $j$, of $2$s assigned. Then the leader sends again $(l,current\_round)$ pairs and waits to receive $j$ different $s$ vectors with equal timestamps. The process continues in such cycles until at some point all existing unique ids report that they didn't manage to assign the current id being assigned.  

\end{document}